%% file: maineptcs.tex
\documentclass[creativecommons]{eptcs}

\providecommand{\event}{GandAlf 2026} 

\usepackage{iftex}

\ifpdf
  \usepackage{underscore}         
  \usepackage[T1]{fontenc}        
\else
  \usepackage{breakurl}           
\fi

\input{preamble}

\newtheorem{theorem}{Theorem}

\newtheorem{remark}{Remark}

\title{The Complexity of Asynchronous HyperLTL}

\author{Gaëtan Regaud\thanks{Supported by the European Union.\newline
\includegraphics[scale=.10]{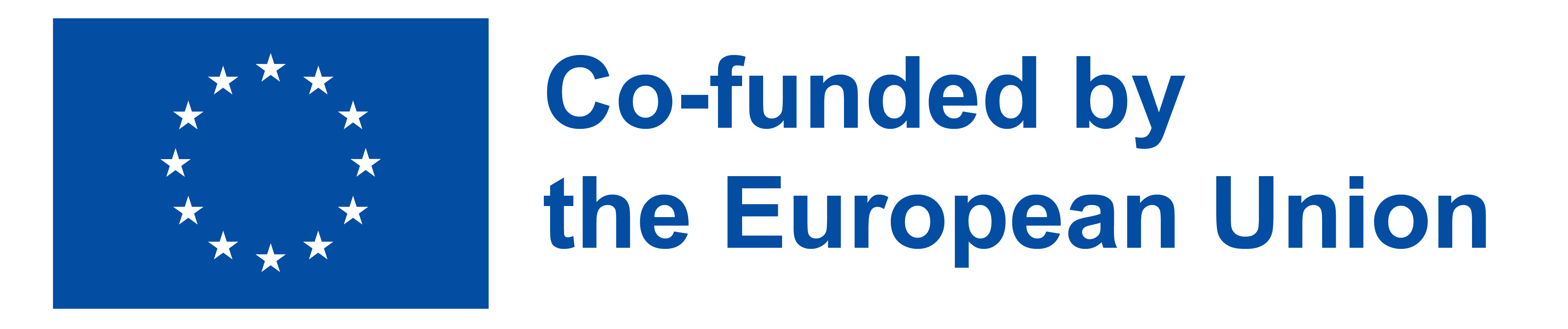}}
\institute{ENS Rennes\\ Rennes, France}
\email{gaetan.regaud@ens-rennes.fr}
\and
Martin Zimmermann\thanks{Partly supported by the project ``Hyperlogics: Expressiveness, Monitorability and Tools (H.-Lo)'' of the Icelandic Research Fund (project no.~2612260-051) and DIREC - Digital Research Centre Denmark.}
\institute{Aalborg University\\ Aalborg, Denmark}
\email{mzi@cs.aau.dk}
}
\def\titlerunning{The Complexity of Asynchronous HyperLTL}

\def\authorrunning{G.\ Regaud and M.\ Zimmermann}

\begin{document}
\maketitle

\begin{abstract}
\input{abstract}
\end{abstract}

\input{content}

\bibliographystyle{eptcs}
\bibliography{bib}

\end{document}

%% file: preamble.tex
\usepackage[utf8]{inputenc}

\usepackage{amsmath}
\usepackage{amssymb}
\usepackage{amsthm}

\usepackage[color=purple!50]{todonotes}

\usepackage{xspace}

\usepackage{colortbl}

\usepackage{booktabs}

\usepackage{mathtools}

\usepackage{tikz}
\usetikzlibrary{arrows,decorations.text,decorations.pathmorphing,decorations.pathreplacing,positioning,patterns,tikzmark,arrows.meta,backgrounds,shapes.geometric}

\tikzset{  
  >=stealth',
  plain/.style 	= {draw, thick,circle, minimum size=8mm, inner sep=0mm}
}

\newcommand{\nats}{\mathbb{N}}
\renewcommand{\epsilon}{\varepsilon}
\renewcommand{\phi}{\varphi}
\newcommand{\size}[1]{|#1|}
\newcommand{\pow}[1]{2^{#1}}
\newcommand{\cceq}{\mathop{::=}}
\newcommand{\set}[1]{\{#1\}}

\newcommand{\F}{\mathop{\mathbf{F}\vphantom{a}}\nolimits}
\newcommand{\G}{\mathop{\mathbf{G}\vphantom{a}}\nolimits}
\DeclareMathOperator{\U}{\mathbf{U}}
\newcommand{\X}{\mathop{\mathbf{X}\vphantom{a}}\nolimits}

\newcommand{\ltl}{{LTL}\xspace}

\newcommand{\ctlstar}{{CTL$^*$}\xspace}
\newcommand{\hyltl}{{Hyper\-LTL}\xspace}
\newcommand{\hypdl}{{Hyper\-PDL-$\Delta$}\xspace}
\newcommand{\hyqptl}{{Hyper\-QPTL}\xspace}
\newcommand{\sohyltl}{{Hyper$^2$LTL}\xspace}
\newcommand{\hmu}{{$H_\mu$}\xspace}

\newcommand{\hyctlstar}{{HyperCTL$^*$}\xspace}

\newcommand{\hyqptlplus}{{Hyper\-QPTL$^+$}\xspace}
\newcommand{\hyaut}{HA\xspace}
\newcommand{\foe}{FO$[E,<]$\xspace}
\newcommand{\hyperfo}{HyperFO\xspace}

\newcommand{\ap}[0]{\mathrm{AP}}

\newcommand{\tower}{\textsc{Tower}\xspace}

\newcommand{\myquot}[1]{``#1''}

\newcommand{\traces}{\mathrm{Tr}}

\newcommand{\prop}{\texttt{p}}

\newcommand{\hyperize}{{\mathit{hyp}}}

\newcommand{\pair}{\mathit{pair}}

\newcommand{\arith}{\mathit{arith}}

\newcommand{\natsstruct}{(\nats, +, \cdot, <, \in)}

\newcommand{\ghltl}{\textrm{GHyLTL$_\textrm{S+C}$}\xspace}

\newcommand{\hltls}{\textrm{HyperLTL$_\textrm S$}\xspace}
\newcommand{\hltlc}{\textrm{HyperLTL$_\textrm C$}\xspace}
\renewcommand{\phi}{\varphi}
\newcommand{\vars}{\textsf{VAR}}
\newcommand{\tra}{\sigma}

\newcommand{\p}{p}

\newcommand{\tr}{x}

\newcommand{\setL}{\mathcal{L}\xspace}

\newcommand{\dom}[1]{\textrm{Dom}(#1)\xspace}

\newcommand{\TS}{\mathcal{T}\xspace}

\newcommand{\labfunc}{\ell}

\newcommand{\auxprop}{\$}

\newcommand{\intprop}{\texttt{\#}}
\newcommand{\altprop}{\dagger}

\newcommand{\ahltl}{\mbox{AHLTL}\xspace}
\newcommand{\eahltl}{\mbox{E-AHLTL}\xspace}
\newcommand{\aahltl}{\mbox{A-AHLTL}\xspace}
\newcommand{\varsocc}[1]{\mathrm{vars}(#1)}
\newcommand{\forward}[4]{\mathit{fwd}_{#2, #3}(#1,#4)}
\newcommand{\onestepforward}[3]{\mathit{os\text{-}fwd}_{#2,#3}(#1)}
\newcommand{\forwardabbr}{\mathit{fwd}}
\newcommand{\tracesabbr}{\mathrm{tr}}
\newcommand{\pointerabbr}{\mathrm{po}}
\newcommand{\osforwardabbr}{\mathrm{os-fwd}}
\newcommand{\E}{\mathsf{E}}
\newcommand{\A}{\mathsf{A}}

\newcommand{\gentquant}{\mathsf{Q}}
\newcommand{\genquant}{\text{\rotatebox[origin=c]{180}{\textsf{Q}}}}

\newcommand{\traceasg}{\Pi_\tracesabbr}
\newcommand{\pointerasg}{\Pi_\pointerabbr}

\newcommand{\init}{\mathrm{init}}

\usepackage{ tipa }

\newcommand{\encode}[1]{\mathit{enc}(#1)}

\newcommand{\propos}{\mathrm{props}}
\newcommand{\varis}{\mathrm{vars}}

%% file: abstract.tex
Hyperproperties express, e.g., information-flow properties of systems, which involves the simultaneous reasoning about multiple execution traces of a system. 
Consequently, HyperLTL, the most important specification logic for hyperproperties, extends LTL with quantification over traces. 
However, HyperLTL can only express synchronous hyperproperties. 

Recently, several logics for asynchronous hyperproperties have been proposed. Here, we focus on AHLTL, asynchronous HyperLTL, which extends HyperLTL with quantification over trajectories that control the relative speed at which time progresses on the quantified traces.
Model-checking AHLTL is known to be undecidable while satisfiability is known to be $\Sigma_1^1$-hard, but the precise complexity of both problems is open. 

Here, we (almost) close these gaps and show that model-checking is equivalent to truth in second-order arithmetic while satisfiability is $\Sigma_1^1$-complete if the trajectory is existentially quantified and $\Sigma_1^1$-hard and in $\Sigma_2^1$ if the trajectory is universally quantified.

%% file: content.tex
\section{Introduction}

The introduction of temporal logics for hyperproperties~\cite{ClarksonFKMRS14} has enabled the specification and automated verification of a wide variety of critical system properties, e.g., from information-flow security, privacy, diagnosability, and system reliability.
Initially, research focused on synchronous hyperproperties as expressed by, e.g., \hyltl and \hyctlstar~\cite{ClarksonFKMRS14}, which extend \ltl and \ctlstar by trace quantification, and thereby express properties of sets of traces.
For example, the \hyltl formula~$
\forall \tr\forall \tr' (\ell_\tr \leftrightarrow \ell_{\tr'}) \rightarrow \G (\ell_\tr \leftrightarrow \ell_{\tr'})
$
expresses observational determinism~\cite{od}, i.e., that any two traces that start with the same low-security observations should at any time have the same low-security observations.

Model-checking is decidable for both logics~\cite{ClarksonFKMRS14,DBLP:conf/cav/FinkbeinerRS15}, which makes them attractive specification languages, witnessed by a plethora of model-checking tools~\cite{BF,DBLP:conf/tacas/BeutnerF23,planning,DBLP:conf/cav/CoenenFST19,DBLP:conf/cav/FinkbeinerRS15,DBLP:conf/atva/GonzalezJSSZ25}.
The complexity of the most important verification problems for synchronous logics, including \hyltl and \hyctlstar, has been thoroughly investigated; see the upper part of Table~\ref{table-knownresults} on Page~\pageref{table-knownresults}.

However, not every system is synchronous. 
Hence, today there are several specification logics for asynchronous hyperproperties, all based on different approaches to asynchronicity:

\begin{itemize}
    \item Asynchronous \hyltl (\ahltl)~\cite{baumeister} adds so-called trajectories to \hyltl, which intuitively specify the rates at which time on different traces evolves. 
    \item \hyltl with stuttering (\hltls)~\cite{DBLP:conf/lics/BozzelliPS21} changes the semantics of the temporal operators of \hyltl so that time does not evolve synchronously on all traces, but instead evolves based on \ltl-definable stuttering.
    \item \hyltl with contexts (\hltlc)~\cite{DBLP:conf/lics/BozzelliPS21} adds a context-operator to \hyltl, which allows one to select a subset of traces on which time progresses synchronously, while it is frozen on all others. 
    \item Generalized \hyltl with stuttering and contexts (\ghltl)~\cite{DBLP:conf/fsttcs/BombardelliB0T24} adds both stuttering and contexts to \hyltl and additionally allows trace quantification under the scope of temporal operators, which \hyltl does not allow. 
    \item \hmu~\cite{hmu} adds trace quantification to the linear-time $\mu$-calculus with asynchronous semantics for the modal operators. 
    \item Hypernode automata (\hyaut)~\cite{hypernode} combine automata and hyperlogic with stuttering. 
\end{itemize}
The known relations between these logics are depicted in Figure~\ref{fig-asynchlogics}.

\begin{figure}[h]
    \centering

    \begin{tikzpicture}[thick,yscale=1.25]
        \node[fill=gray!20,rounded corners] (hyltl) at (1,0.5) {\hyltl};
        \node[fill=gray!20,rounded corners,align=center,anchor = north] (shltls) at (4,2) {simple\\ \hltls};
        \node[fill=gray!20,rounded corners,align=center,anchor = north] (hltls) at (4,3) {\hltls};
        \node[fill=gray!20,rounded corners,align=center,anchor = north] (hltlc) at (-2,3) {\hltlc};
        \node[fill=gray!20,rounded corners,align=center,anchor = north] (sghltl) at (1,3) {simple\\ \ghltl};
        \node[fill=gray!20,rounded corners,align=center] (ghltl) at (1,4) {\ghltl};

        \node[fill=gray!20,rounded corners,align=center,anchor = north] (foe) at (6.5,2) {\foe};
        \node[fill=gray!20,rounded corners,align=center] (hyperfo) at (6.5,0.5) {\hyperfo};
        
        \node[fill=gray!20,rounded corners,align=center,anchor = north] (hyaut) at (8.5,2) {\hyaut};
        \node[fill=gray!20,rounded corners,align=center,anchor = north] (ahyltl) at (-4,2) {\ahltl};
        \node[fill=gray!20,rounded corners,align=center] (hmu) at (-4,4) {\hmu};

        \path[->, > = stealth]
        (hyltl) edge node[above] {\footnotesize \cite{expressiveness}} (ahyltl.south)
        (hyltl) edge node[above] {} (hltlc)
        (hyltl) edge node[above] {} (shltls)
        (shltls) edge node[above] {} (hltls)
        (shltls) edge[bend left] node[above] {} (sghltl) 
        (sghltl) edge node[above] {} (ghltl)
        (hltls) edge[bend right] node[above] {} (ghltl.east)
        (hltlc) edge[bend left] node[above] {} (ghltl.west)
        (hyltl) edge[<->] node[below] {\footnotesize \cite{FZ17}} (hyperfo)
        (hyperfo) edge node[above] {} (foe)
        (ahyltl) edge node[left ]{\footnotesize \cite{expressiveness}} (hmu)
        (foe) edge[bend right] node[above,yshift=10] {\footnotesize \cite{DBLP:conf/fsttcs/BombardelliB0T24}} (ghltl)
        (hltls.north west) edge[bend right=5] node[above, near end]{\footnotesize \cite{DBLP:conf/lics/BozzelliPS21}} (hmu)
        (hltlc) edge node[above]{\footnotesize \,\,\cite{DBLP:conf/lics/BozzelliPS21}} (hmu)
        ;

        \draw[dashed, rounded corners] (-5,1) -- (-.5,1) -- (-.5,3.25) -- (2.55, 3.25) -- (2.55, 2.2) -- (7.75,2.2) -- (7.75, 1) -- (9.25,1);
        
    \end{tikzpicture}
    
    \caption{The landscape of logics for asynchronous hyperproperties.
    Here, \foe is first-order logic with the equal-level predicate~$E$ and order~$<$ evaluated over sets of infinite words, and \hyperfo is a fragment that is equivalent to \hyltl~\cite{FZ17}.
    Arrows denote known inclusions (where the inclusion of \hyltl into \ahltl requires an additional proposition) and the dashed line denotes the decidability border for model-checking. For non-inclusions, we refer the reader to work on the expressiveness of asynchronous hyperlogics~\cite{expressiveness,DBLP:conf/fsttcs/BombardelliB0T24}.
    }
    \label{fig-asynchlogics}
\end{figure}

Unlike for synchronous hyperlogics, the complexity of many fundamental verification problems for asynchronous hyperlogics is still open. 
A first step has recently been made for generalized \hyltl with stuttering and contexts and its fragments, e.g., \hltls and \hltlc; see the middle part of Table~\ref{table-knownresults} for an overview.
However, for the other logics the picture is wide open: typically model-checking is undecidable for the full logic, but decidable for fragments, and satisfiability is $\Sigma_1^1$-hard if \hyltl can be embedded.

\paragraph{Our Contribution.}
Here, we study the complexity of model-checking, satisfiability, and finite-state satisfiability for \ahltl. 
In \ahltl, one explicitly quantifies a trajectory that controls how time progresses on the different traces quantified before, i.e., asynchronicity is semantic. 
For example, the \ahltl formula~$
\forall \tr\forall \tr'\E (\ell_\tr \leftrightarrow \ell_{\tr'}) \rightarrow \G (\ell_\tr \leftrightarrow \ell_{\tr'})$
expresses observational determinism for asynchronous systems, i.e., for any two traces there must be a trajectory such that if the traces start with the same low-security observations, then the states visited under the stuttering induced by the trajectory have the same low-security observations.

We show that for existentially quantified trajectories (\eahltl), satisfiability is not harder than for \hyltl, which can also be embedded into \eahltl~\cite{expressiveness}.
Thus, \eahltl satisfiability is $\Sigma_1^1$-complete.
As the lower bound follows trivially from previous work, we only need to prove the upper bound. 
Here, we show that every formula of \ahltl (i.e., independently of how the trajectory is quantified) has a countable model. The existence of such a model, Skolem functions for it, and a trajectory, which are all second-order objects, and their correctness can then be expressed using a $\Sigma_1^1$ formula. 
For universally quantified trajectories, we use the same approach and obtain a $\Sigma_2^1$-upper bound, as we can still existentially quantify the model and Skolem functions, but need universal second-order quantification to capture the trajectory quantification.
For the lower bound, we show that \hyltl can also be embedded into \aahltl, i.e., \aahltl satisfiability is $\Sigma_1^1$-hard.

Then, we show that model-checking \aahltl is as hard as truth in second-order arithmetic, i.e., much harder than satisfiability. Intuitively, the reason is that one has to deal with uncountable models in model-checking, while we have proven that countable models suffice for satisfiability. 
From the lower bound for \aahltl model-checking, we immediately also obtain the same lower bound for \eahltl model-checking.
This is complemented by showing that model-checking \ahltl (i.e., independently of how the trajectory is quantified) is reducible to truth in second-order arithmetic, yielding a matching upper bound for the full logic \ahltl.

Finally, as for other hyperlogics with highly undecidable model-checking (see, e.g., second-order \hyltl~\cite{frz26} and \hyqptlplus~\cite{hyperqptlcomplexity}), model-checking for \ahltl (and its fragments) is inter-re\-ducible with finite-state satisfiability, i.e., we obtain the same complexity for the latter problem.

Our results are listed in the lower third of Table~\ref{table-knownresults}.

\begin{table}[h]
    \setlength{\tabcolsep}{4pt}
    \centering
    \caption{List of complexity results for synchronous hyperlogics, fragments of \ghltl, and our results. \myquot{T2A-equiv.} (\myquot{T3A-equiv.}) stands for \myquot{equivalent to truth in second-order (third-order) arithmetic}. The result for \hypdl satisfiability can be shown using techniques developed  by Fortin et al.\ for \hyltl satisfiability~\cite{hypercomplexity}, the result for \hypdl finite-state satisfiability follows from the lower bound for its fragment~\hyltl and \hypdl model-checking being decidable.}
    
    \renewcommand{\arraystretch}{1.1}

    \begin{tabular}{llll}
    
       Logic &  Satisfiability  & Finite-state Satisfiability  &  Model-checking \\
         \midrule
    
        \hyltl & $\Sigma_1^1$-compl.~\cite{hypercomplexity}  & $\Sigma_1^0$-compl.~\cite{FinkbeinerH16} & \tower-compl.~\cite{Rabe16diss,MZ20} \\
        
        \rowcolor{lightgray!40} \hypdl & $\Sigma_1^1$-compl.  & $\Sigma_1^0$-compl. & $\tower$-compl.~\cite{hyperpdl}\\
        
        \hyqptl &  $\Sigma_1^2$-complete.~\cite{hyperqptlcomplexity} & $\Sigma_1^0$-compl.~\cite{Rabe16diss} & $\tower$-compl.~\cite{Rabe16diss}\\
        
        \rowcolor{lightgray!40}\hyqptlplus & T3A-equiv.~\cite{hyperqptlcomplexity} & T3A-equiv.~\cite{hyperqptlcomplexity} &   T3A-equiv.~\cite{hyperqptlcomplexity}\\
        
        \sohyltl & T3A-equiv.~\cite{frz26}  & T3A-equiv.~\cite{frz26}  &  T3A-equiv.~\cite{frz26}\\
        
        \rowcolor{lightgray!40}\hyctlstar & $\Sigma_1^2$-compl.~\cite{hypercomplexity}  & $\Sigma_1^0$-compl.~\cite{FinkbeinerH16} & \tower-compl.~\cite{Rabe16diss,MZ20} \\
    
         \midrule
    
\ghltl & $\Sigma_1^1$-compl.~\cite{genhyltlsc} & T2A-equiv.~\cite{genhyltlsc} & T2A-equiv.~\cite{genhyltlsc} \\
\rowcolor{lightgray!40}\hltls & $\Sigma_1^1$-compl.~\cite{genhyltlsc} & T2A-equiv.~\cite{genhyltlsc} & T2A-equiv.~\cite{genhyltlsc} \\
\hltlc & $\Sigma_1^1$-compl.~\cite{genhyltlsc} & T2A-equiv.~\cite{genhyltlsc} & T2A-equiv.~\cite{genhyltlsc} \\
\rowcolor{lightgray!40}simple \ghltl & $\Sigma_1^1$-compl.~\cite{genhyltlsc} & $\Sigma_1^0$-compl.~\cite{genhyltlsc} & \tower-compl.~\cite{DBLP:conf/fsttcs/BombardelliB0T24,Rabe16diss} \\
simple \hltls & $\Sigma_1^1$-compl.~\cite{genhyltlsc} & $\Sigma_1^0$-compl.~\cite{genhyltlsc} & \tower-compl.~\cite{DBLP:conf/fsttcs/BombardelliB0T24,Rabe16diss} \\

        \midrule
    
     \rowcolor{lightgray!40}   \ahltl & $\Sigma_1^1$-hard/in $\Sigma_2^1$ (Thms.~\ref{thm-Esat},\ref{thm-Asat})  & T2A-equiv. (Thm.~\ref{thm-fssat}) & T2A-equiv. (Thm.~\ref{thm-mc}) \\
        
         \aahltl & $\Sigma_1^1$-hard/in $\Sigma_2^1$ (Thm.~\ref{thm-Asat})  & T2A-equiv. (Thm.~\ref{thm-fssat}) & T2A-equiv. (Thm.~\ref{thm-mc})\\
        
   \rowcolor{lightgray!40}     \eahltl &  $\Sigma_1^1$-complete. (Thm.~\ref{thm-Esat}) & T2A-equiv.~(Thm.~\ref{thm-fssat}) & T2A-equiv.~(Thm.~\ref{thm-mc})\\

    \end{tabular}
    \label{table-knownresults}
\end{table}

\section{Preliminaries}

The set of nonnegative integers is denoted by $\nats$. 
We fix a countable set~$\ap$ of atomic propositions.

An alphabet is a nonempty finite set~$\Sigma$.
The set of infinite words over $\Sigma$ is denoted by $\Sigma^\omega$. Given $w \in \Sigma^\omega$ and $i \in\nats$, $w(i)$ denotes the $i$-th letter of $w$ (starting with $i =0$).
A trace~$\tra$ over $\ap$ is an element of $(\pow\ap)^\omega$.

A transition system is a tuple~$\TS=(V,E,I,\labfunc)$ where $V$ is a nonempty finite set of vertices, $E\subseteq V\times V$ is a set of directed edges, $I\subseteq V$ is a nonempty set of initial vertices, and $\labfunc: V\to \pow\ap$ is a labeling function that maps each vertex to a \emph{finite} set of propositions. We require that each vertex has at least one outgoing edge.
A run of $\TS$ is an infinite word~$v_0v_1\cdots\in V^\omega$ such that $v_0\in I$ and $(v_i,v_{i+1})\in E$ for all $i\in\nats$. The set~$\traces(\TS)=\{\labfunc(v_0)\labfunc(v_1)\cdots\mid v_0v_1\cdots\text{is a run of $\TS$}\}$ is the set of traces induced by $\TS$.

\subsection{Asynchronous HyperLTL}

Let $\vars$ be a countable set of (trace) variables. 
The syntax of \ahltl is given by the grammar
\[
\phi \cceq \exists\tr\phi\mid\forall\tr\phi \mid \E\psi \mid \A\psi  \qquad\qquad
\psi \cceq  \p_\tr \mid \lnot\psi\mid\psi\lor\psi\mid\X\psi\mid\psi\U\psi
\]
where $\p\in\ap$ and $\tr\in\vars$.
We use the usual syntactic sugar, e.g., $\land$, $\rightarrow$, $\leftrightarrow$, $\F$ (eventually), and $\G$ (always).
We denote the set of trace variables occurring as subscripts of propositions in $\phi$ by $\varsocc{\phi}$.
A formula is a sentence, if it has no free trace variables, which are defined as expected.
Note that each sentence must either contain the existential trajectory quantifier~$\E$ or the universal trajectory quantifier~$\A$.
We call the corresponding fragments (of sentences)~\eahltl and \aahltl.

To define the semantics of \ahltl, we need to introduce some more notation.
A trace assignment is a partial function~$\traceasg\colon\vars\to (\pow\ap)^\omega$ that maps trace variables to traces.
For $\tr\in\vars$ and $\tra\in(\pow\ap)^\omega$, the trace assignment~$\traceasg[\tr\mapsto \tra]$ maps $\tr$ to $\tra$ and each other~$\tr' \in\dom{\traceasg} \setminus\set{\tr}$ to $\traceasg(\tr')$.
A pointer assignment is a partial function~$\pointerasg \colon \vars \to \nats$ that maps trace variables to positions. 
For $I \subseteq \vars$, the pointer assignment~$\pointerasg+I$ maps each $\tr$ in the domain of $\pointerasg$ to $\pointerasg(\tr) +1$ if $x \in I$ and to $\pointerasg(\tr)$ otherwise.
Intuitively, $\pointerasg(\tr)$ will denote the current time on $\traceasg(\tr)$.

A trajectory for a quantifier-free formula~$\psi$ is an infinite sequence~$t(0)t(1)t(2)\cdots$ of nonempty subsets of $\varsocc{\psi}$.
Given a pointer assignment~$\pointerasg$, a trajectory~$t$, and $i \in \nats$, we define the one-step forward operator as $\onestepforward{\pointerasg}{t}{i} = \pointerasg + t(i)$
and the iterated forward operator as $\forward{\pointerasg}{t}{i}{0} = \pointerasg$ and 
\[\forward{\pointerasg}{t}{i}{k+1} = 
\forward{
\onestepforward{\pointerasg}{t}{i}
}
{t}
{i+1}
{k}
\]
for $k \ge 0$.
Intuitively, $\forward{\pointerasg}{t}{i}{k}$ increments each pointer~$\pointerasg(\tr)$ by the cardinality of the set
\[
\set{j \in \set{i, i+1, \ldots, i+k-1} \mid \tr \in t(j) },
\]
i.e., the trajectory determines how time progresses on the level of individual traces.

The semantics\footnote{Our definition of the semantics differs syntactically from the original one by Baumeister et al.~\cite{baumeister}, e.g., we are updating pointers on the traces when evaluating temporal operators instead of updating each trace with a suitable suffix. These differences will simplify our proofs considerably. Nevertheless, both definitions are equivalent. 
} of \ahltl is defined for a set~$\setL$ of traces, a trace assignment~$\traceasg$, a pointer assignment~$\pointerasg$, a trajectory~$t$, and a position~$i$ on the trajectory as follows (note that not all objects are needed for all kinds of formulas):
\begin{itemize}
    \item $(\setL, \traceasg) \models \exists \tr\ \phi$ if there exists a $\tra \in \setL$ such that $(\setL, \traceasg[\tr \mapsto \tra]) \models \phi$,

    \item $(\setL, \traceasg ) \models \forall \tr\ \phi$ if for all $\tra \in \setL$ we have $(\setL, \traceasg[\tr \mapsto \tra]) \models \phi$,

    \item $(\setL, \traceasg ) \models \E \psi$ if there exists a trajectory~$t$ for $\psi$ such that $(\traceasg, \pointerasg^\init, t, 0) \models \psi$, where $\pointerasg^\init$ is the pointer assignment mapping all variables to $0$,

    \item $(\setL, \traceasg ) \models \A \psi$ if for all trajectories~$t$ for $\psi$ we have $(\traceasg, \pointerasg^\init, t, 0) \models \psi$,

    \item $(\traceasg, \pointerasg, t,i) \models \p_\tr$ if $\traceasg(\tr) = \tra$ and $\p \in \tra(\pointerasg(\tr))$,
    
    \item $(\traceasg, \pointerasg, t,i) \models \lnot \psi$ if $(\traceasg, \pointerasg, t,i) \not\models \psi$,

    \item $(\traceasg, \pointerasg, t,i) \models \psi_0 \lor \psi_1$ if $(\traceasg, \pointerasg, t,i) \models \psi_0$ or $(\traceasg, \pointerasg, t,i) \models \psi_1$,
    
    \item $(\traceasg, \pointerasg, t,i) \models \X \psi$ if $(\traceasg, \forward{\pointerasg}{t}{i}{1}, t,i+1) \models \psi$, and

    \item $(\traceasg, \pointerasg, t,i) \models \psi_0 \U \psi_1$ if there is a $k \ge 0$ such that $(\traceasg, \forward{\pointerasg}{t}{i}{k}, t,i+k) \models \psi_1$ and\\ $(\traceasg, \forward{\pointerasg}{t}{i}{k'}, t,i+k') \models \psi_0$ for all $0 \le k' < k$.
        
\end{itemize}

\begin{remark}
\label{remark-boolcombs}
While \ahltl sentences are required to be in prenex normal form (with the trajectory quantifier being the last one), the fragments \eahltl and \aahltl are closed under conjunction and disjunction, e.g., $(\forall \tr_0 \exists \tr_1 \A \psi) \vee (\exists \tr_2 \forall \tr_3 \A \psi')$ is equivalent to $\forall \tr_0 \exists \tr_1 \exists \tr_2 \forall \tr_3 \A \psi \vee \psi'$.
Here, as usual, we just have to assume, w.l.o.g., that the sentences all use pairwise disjoint sets of trace variables.
\end{remark}

We say that a nonempty set~$\setL$ of traces satisfies a sentence~$\phi$, if $(\setL, \traceasg^\emptyset) \models \phi$, where
$\traceasg^\emptyset$ denotes the trace assignment with empty domain.
A transition system~$\TS$ satisfies $\phi$, written $\TS \models \phi$, if $\traces(\TS) \models \phi$.

\begin{remark}[cf.\ Section~3.1 of \cite{expressiveness}]
\label{remark-trajaccess}
While there is no direct way for accessing the trajectory using \ahltl formulas, one can still indirectly access it.

Let $t$ be a trajectory, $\traceasg$ be a trace assignment,  and let $\altprop$ be a proposition. 
If $(\traceasg, \pointerasg^\init, t, i) \models \altprop_\tr \leftrightarrow \neg\X \altprop_\tr $ then $\tr \in t(i)$.
Hence, if $(\traceasg, \pointerasg^\init, t, 0) \models \G(\altprop_\tr \leftrightarrow \neg\X \altprop_\tr )$ then $\tr \in t(i)$ for all $i$.

Now assume $\traceasg(\tr) = \tra$ satisfies $\altprop \in \tra(n)$ if and only if $\altprop \notin \tra(n+1)$ for all $n$, i.e., the truth value of $\altprop$ alternates in $\tra$. 
Then, $(\traceasg, \pointerasg^\init, t, i) \models \altprop_\tr \leftrightarrow \neg \X \altprop_\tr $ if and only if $\tr \in t(i)$.
Hence, $(\traceasg, \pointerasg^\init, t, 0) \models \G(\altprop_\tr \leftrightarrow \neg \X \altprop_\tr) $ if and only if $\tr \in t(i)$ for all $i$.
\end{remark}

We are interested in satisfiability, finite-state satisfiability, and model-checking for \ahltl:
\begin{itemize}
    \item Satisfiability: Given a sentence~$\phi$ of \ahltl, is there a nonempty set~$\setL$ of traces with $\setL \models \phi$?

    \item Finite-state satisfiability: Given a sentence~$\phi$ of \ahltl, is there a transition system~$\TS$ (which is finite by definition) such that $\TS \models \phi$?

    \item Model-checking: Given a transition system~$\TS$ and a sentence~$\phi$ of \ahltl, do we have $\TS \models \phi$?
\end{itemize}

\subsection{Arithmetic and Complexity Classes for Undecidable Problems.}

To capture the complexity of undecidable problems, we consider formulas of arithmetic, i.e., predicate logic with signature~$(+, \cdot, <, \in)$, evaluated over the structure~$\natsstruct$. 
A type~$0$ object is a natural number in $\nats$ and a type~$1$ object is a subset of $\nats$.
In the following, we use lower-case roman letters (possibly with decorations) for first-order variables, and upper-case roman letters (possibly with decorations) for second-order variables.
Every fixed natural number is definable in first-order arithmetic, so we freely use them as syntactic sugar. For more detailed definitions, we refer to \cite{Rogers87}.

Our benchmark is second-order arithmetic, i.e., predicate logic with quantification over type~$0$ and type~$1$ objects. 
Arithmetic formulas with a single free first-order variable define sets of natural numbers. In particular, $\Sigma_1^1$ contains the sets of the form
\[\set{x \in\nats \mid \exists X_1 \subseteq \nats , \ldots, \exists X_k\subseteq \nats \text{ such that } \natsstruct\models \psi(x, X_1, \ldots,X_k )},\] where $\psi$ is a formula of arithmetic with arbitrary quantification over type~$0$ objects (but no second-order quantifiers). 
Similarly, $\Sigma_2^1$ contains the sets of the form
\begin{multline*}
\set{x \in\nats \mid \exists X_1 \subseteq \nats , \ldots, \exists X_k\subseteq \nats,
\forall X_1' \subseteq \nats, \ldots, \forall X_{k'}'\subseteq \nats \text{ we have }\\
\natsstruct\models \psi(x, X_1, \ldots,X_k,X_1', \ldots,X_{k'}' )},    
\end{multline*}
where $\psi$ is a formula of arithmetic with arbitrary quantification over type~$0$ objects (but no second-order quantifiers).
Furthermore, truth in second-order arithmetic is the following problem: Given a sentence~$\phi$ of second-order arithmetic, do we have $\natsstruct\models\phi$?

\section{Satisfiability}

In this section, we study the satisfiability problem for \ahltl. 
Our upper bounds rely on a \myquot{small} model property, which is proven similarly to the analogous result for \hyltl.

\begin{theorem}
\label{thm-modelsize}
Every satisfiable \ahltl sentence has a countable model.
\end{theorem}

\begin{proof}
Let $\phi = \genquant_0 \tr_0 \genquant_1 \tr_1 \cdots \genquant_{k-1}\tr_{k-1}\gentquant \psi$ be a satisfiable \ahltl sentence, where $\genquant_j \in \set{\exists,\forall}$ for all $j$ and $\gentquant \in \set{\E,\A}$.
Note that this implies that $\psi$ is quantifier-free.
As $\phi$ is satisfiable, there is a nonempty~$\setL$ such that $\setL\models \phi$. We construct a countable subset~$\setL_\omega \subseteq \setL$  with $\setL_\omega \models \phi$.

For every existentially quantified trace variable~$\tr$ in $\phi$, let $U_{\tr}$ be the set of trace variables quantified universally before $\tr$.
For every such $\tr$, say with $U_{\tr} = \set{ \tr_{j_1},\ldots, \tr_{j_{\size{U_{\tr}}}}}$, there exists a Skolem function~$f_{\tr}\colon \setL ^{\size{U_{\tr}}}\rightarrow \setL $ such that for each trace assignment~$\traceasg$ mapping universally quantified variables to traces in $\setL$ and each existentially quantified variable~$\tr$ to $f_{\tr}(\traceasg(\tr_{j_1}),\ldots, \traceasg(\tr_{j_{\size{U_{\tr}}}}))$, we have $(\setL, \traceasg)\models \gentquant\psi$.
Note that $(\setL, \traceasg)\models \gentquant\psi$ is independent of $\setL$, as $\gentquant\psi$ does not contain trace quantifiers.

We fix such Skolem functions~$f_{\tr}$ for the rest of the proof. For a set~$\setL'\subseteq \setL $ of traces and an existentially quantified variable~$\tr$, we define 
\[f_{\tr}(\setL')=\{f_{\tr}(\tra_{j_1},\ldots, \tra_{j_{\size{U_{\tr}}}})\mid \tra_{j_1},\ldots , \tra_{j_{\size{U_{\tr}}}} \in \setL'\}.
    \]
    Note that if $\setL'$ is finite, then $f_{\tr}(\setL ')$ is also finite.
    Now, define $\setL_0=\{\tra\}$ for some arbitrary $\tra\in \setL $ and $\setL_{i+1}=\setL_i\cup\bigcup_{\tr} f_{\tr}(\setL_i)$ for all $i\in\nats$, where $\tr$ ranges over the existentially quantified variables in $\phi$.
    Then, $\setL_\omega=\bigcup_{i\in\nats} \setL_i$ is a countable union of finite sets and therefore countable.
    Thus, it remains to show that it is also a model of $\phi$.
    
    By definition of Skolem functions, every trace assignment~$\traceasg$ that maps universally quantified variables to traces in $\setL_\omega$ and uses the Skolem functions for existentially quantified variables satisfies $(\setL,\traceasg)\models\gentquant\psi$. Thus, we also have $(\setL_\omega,\traceasg)\models\gentquant\psi$, as there are no trace quantifiers in $\gentquant\psi$.
    Now, an induction over the quantifier prefix of $\phi$ (from the inside out) shows that $(\setL_\omega,\traceasg)\models \genquant_j\tr_j\ldots \genquant_{k-1}\tr_{k-1}\gentquant\psi$ for all $j$, i.e., $\setL_\omega$ is a model of $\phi$.
    Here, we rely on the fact that $\setL_\omega$ is closed under the application of the Skolem functions, i.e., if $\tra_{j_1},\ldots , \tra_{j_{\size{U_{\tr}}}} \in \setL_\omega$, then $f_{\tr}(\tra_{j_1},\ldots, \tra_{j_{\size{U_{\tr}}}}) \in \setL_\omega$.
\end{proof}

Recall that every \ahltl sentence is either a \eahltl sentence (if the trajectory is quantified existentially) or an \aahltl sentence (if the trajectory is quantified universally).
We treat both fragments individually. 

\begin{theorem}
\label{thm-Esat}    
\eahltl satisfiability is $\Sigma_1^1$-complete.
\end{theorem}

\begin{proof}
We begin with the upper bound. 
Due to Theorem~\ref{thm-modelsize}, we can restrict ourselves to countable models, i.e., an \eahltl-sentence is satisfiable if and only if it is satisfied by a countable set of traces. 
In the following, we show how to express the existence of a countable model in arithmetic using only existential second-order quantification and arbitrary first-order quantification.

So, let us fix an \eahltl-sentence~$\phi$.
We assume w.l.o.g.\ that both the variables appearing in $\phi$ and the propositions appearing in $\phi$ are natural numbers.
More specifically, we can assume that $\phi$ uses the variables~$\set{0,1,\ldots, \ell-1}$ for some $\ell > 0$.
For technical convenience we will restrict ourselves in this proof to variable assignments with domain~$\set{0,1,\ldots, \ell-1}$, as this is sufficient to capture the semantics.

We begin by explaining how to encode the different objects we need to handle to capture the semantics of \ahltl using natural numbers and sets of natural numbers (equivalently, functions from $\nats^m$ to $\nats^n$ for some $n$ and $m$).
For example, we can \emph{name} the traces of a nonempty countable set of traces by natural numbers (if the set is only finite, the naming is not injective, which is inconsequential for our construction).
So, we can encode a nonempty countable set of traces by a function mapping natural numbers (i.e., trace names), positions on traces, and propositions to $\set{0,1}$ indicating whether a proposition holds at a position of a trace, i.e., such a set of traces is encoded by a function from $\nats^3$ to $\nats$, i.e., by a type~$1$ object.

Next, we consider trace and pointer assignments:
We can encode a trace assignment by a list of natural numbers of length~$\ell$ containing the names of the traces the variables are mapped to.
Similarly, we can encode a pointer assignment by a list of natural numbers of length~$\ell$ containing the positions the variables are mapped to.
Hence, both are type~$0$ objects, as finite lists of numbers can be encoded by numbers.
Using this encoding, we can encode a Skolem function for an existentially quantified variable by a function mapping natural numbers (encoding trace assignments) to natural numbers (encoding trace names), which is again a type~$1$ object.
Finally, a trajectory maps natural numbers to sets of variables.
We encode each such set by a finite list of natural numbers, which is again a type~$0$ object.
So, a trajectory is encoded by a type~$1$ object.
All these encodings can be \myquot{implemented} in first-order arithmetic.

Thus, we can intuitively use existential second-order quantification to express the existence of (the encoding of) a countable model, of (the encoding of) Skolem functions for the existentially quantified variables, and of (the encoding of) a trajectory.
Thus, it remains to capture that every trace assignment (a type~$0$ object) that is consistent with the Skolem functions satisfies, with respect to the trajectory, the maximal quantifier-free subformula~$\psi$ of $\phi$, when all pointers are zero.

To this end, let $\Psi$ be the set of subformulas of $\psi$, let $t$ be a trajectory, and let $\traceasg$ be a trace assignment. 
The $(\psi, \traceasg, t)$-expansion is the function~$e_{\psi, \traceasg, t}$ mapping a pointer assignment~$\pointerasg$, a position~$i \in \nats$ for the trajectory~$t$, and a subformula~$\psi' \in \Psi$ to
\[
e_{\psi, \traceasg, t}(\pointerasg, i,\psi') = 
\begin{cases}
    1 & \text{if } (\traceasg, \pointerasg, t,i) \models \psi', \\
    0 & \text{otherwise.}
\end{cases}
\]

The $(\psi, \traceasg, t)$-expansion is uniquely characterized by the following consistency requirements:
\begin{itemize}
    \item $e_{\psi, \traceasg, t}(\pointerasg, i,\prop_\tr) = \begin{cases}
        1 &\text{if } \traceasg(\tr) = \tra \text{ and } \p \in \tra(\pointerasg(\tr)),\\
        0 &\text{otherwise,}
    \end{cases}$
    \item $e_{\psi, \traceasg, t}(\pointerasg, i, \neg \psi')= 1$ if and only if $e_{\psi, \traceasg, t}(\pointerasg, i,\psi')= 0$,
    \item $e_{\psi, \traceasg, t}(\pointerasg, i,\psi_0 \vee \psi_1) = 1$ if and only if $e_{\psi, \traceasg, t}(\pointerasg, i,\psi_0) =1$ or $e_{\psi, \traceasg, t}(\pointerasg, i,\psi_1) = 1$,
    \item $e_{\psi, \traceasg, t}(\pointerasg, i,\X\psi') = 1$ if and only if $e_{\psi, \traceasg, t}(\forward{\pointerasg}{t}{i}{1}, i+1,\psi') = 1$, and
    \item $e_{\psi, \traceasg, t}(\pointerasg, i,\psi_0 \U \psi_1) = 1$ if and only if there is a $k\ge 0$ such that $e_{\psi, \traceasg, t}(\forward{\pointerasg}{t}{i}{k}, i+k,\psi_1) = 1$ and $e_{\psi, \traceasg, t}(\forward{\pointerasg}{t}{i}{k'}, i+k',\psi_0) = 1$ for all $0 \le k' < k$.
\end{itemize}

We argue that these consistency requirements can be captured in arithmetic. 
Then, we can existentially quantify a type~$1$ object encoding the $(\psi, \traceasg, t)$-expansion for a given $\traceasg$ and $t$ and thus check whether $\traceasg$ and $t$ satisfy $\psi$.

To this end, we first argue that one can \myquot{implement} the (one-step) forward operator in arithmetic, as all other requirement can directly be expressed in first-order logic.

We can write a formula~$\alpha_\osforwardabbr(T, i, a_p, a_p')$ with free variables~$T$ (second-order, encoding a trajectory~$t$), $i$ (first-order, interpreted as a position on $t$), and $a_p$ and $a_p'$ (first-order, each encoding a pointer assignment~$\pointerasg$ and $\pointerasg'$) that is satisfied in $\natsstruct$ if and only if $\onestepforward{\pointerasg}{t}{i} = \pointerasg'$.
Then, we can write a formula~$\alpha_\forwardabbr(T, i, a_p, a_p', k)$ with free variables~$T$ (second-order, encoding a trajectory~$t$), $i$ (first-order, interpreted as a position on $t$), $a_p$ and $a_p'$ (first-order, each encoding a pointer assignment~$\pointerasg$ and $\pointerasg'$), and $k$ (first-order) that is satisfied in $\natsstruct$ if and only if $\forward{\pointerasg}{t}{i}{k} = \pointerasg'$:
The formula expresses that there exists a list~$(a_0, a_1, \ldots, a_k)$ of length~$k+1$ such that $a_0= a_p$, $a_k = a_p'$, and $\alpha_\forwardabbr(T, i+n, a_n, a_{n+1})$ holds for all $0 \le n < k$.

Now, we can complete the proof: We can construct a formula~$\alpha(x)$ of arithmetic with a single free first-order variable~$x$ so that we have for all \eahltl sentences~$\phi$: $\natsstruct \models \alpha(\encode{\phi})$ if and only if $\phi$ is satisfiable. Here, $\encode{\phi}$ is some suitable encoding of \eahltl sentences by natural numbers, which can be implemented in first-order arithmetic.
The formula~$\alpha$ expresses, for a given~$\encode{\phi}$, that there is (the encoding of) a countable model, there are (encodings of) Skolem functions for the existentially quantified variables in $\phi$, there is (the encoding of) a trajectory~$t$, and there is a function~$E$ from $\nats^4 $ to $\nats$ such that for all trace assignments~$\traceasg$ (encoded as some~$a_t \in \nats$) that are consistent with the Skolem functions, the function~$a_p,i,p \mapsto E(a_t,a_p,i,p)$ is the $(\psi, \traceasg,t)$-expansion and we have $E(a_t, a_0, 0, p^*) = 1$.
Here, $a_p$ encodes a pointer assignment, $p$ encodes a quantifier-free subformula of $\phi$, and $p^*$ is the encoding of the maximal quantifier-free subformula of $\phi$.
We leave the tedious, but standard, details to the reader.

The matching lower bound follows from previous work by Bozzelli, Peron, and Sánchez who presented a satisfiability-preserving embedding of \hyltl in \eahltl~\cite[Theorem 3]{expressiveness}.
Thus, as \hyltl satisfiability is $\Sigma_1^1$-hard~\cite{hypercomplexity}, so is \eahltl satisfiability.

As we adapt their embedding below for \aahltl satisfiability, let us briefly recall it here: 
Let $\varphi = \genquant_0 \tr_0 \ldots \genquant_{k-1}\tr_{k-1}\psi$ be a \hyltl sentence with quantifier-free $\psi$ and let $\altprop$ be a proposition not appearing in $\phi$. Then, due to Remark~\ref{remark-trajaccess}, $\phi$ is satisfiable if and only if 
\[\genquant_0 \tr_0 \ldots \genquant_{k-1}\tr_{k-1} \E \bigwedge_{j=0}^{k-1} \G(\altprop_{\tr_j} \leftrightarrow \neg \X\altprop_{\tr_j}) \wedge \psi\]
is satisfiable, i.e., we existentially quantify a trajectory~$t(0)t(1)t(2)\cdots$ and then require that each $\tr_j$ is in each~$t(i)$, i.e., time on all traces quantified in $\phi$ progresses normally. 
Then, the semantics of the temporal operators in \ahltl and \hyltl coincide.

\end{proof}

Next, we consider satisfiability for formulas with universally quantified trajectories. 
A $\Sigma_1^1$ lower bound can again be obtained by embedding \hyltl in \aahltl.
However, recall that we existentially quantified a model, Skolem functions, the expansion, and a trajectory to obtain the matching $\Sigma_1^1$ upper bound for \eahltl. 
However, for \aahltl, the trajectory is universally quantified, i.e., we obtain a $\Sigma_2^1$ upper bound, as we existentially quantify a model, Skolem functions, and the expansion, and then universally quantify the trajectory.

\begin{theorem}
\label{thm-Asat}  
\aahltl satisfiability is $\Sigma_1^1$-hard and in $\Sigma_2^1$.
\end{theorem}

\begin{proof}
For the lower bound, we adapt the embedding of \hyltl into \eahltl~\cite{expressiveness} recalled in the proof of Theorem~\ref{thm-Esat}.
Thus, as \hyltl satisfiability is $\Sigma_1^1$-hard~\cite{hypercomplexity}, so is \aahltl satisfiability. 

Let $\varphi = \genquant_0 \tr_0 \ldots \genquant_{k-1}\tr_{k-1}\psi$ be a \hyltl sentence with quantifier-free $\psi$ and let $\altprop$ be a proposition not appearing in $\phi$. 
Due to Remark~\ref{remark-trajaccess}, $\phi$ is satisfiable if and only if \[\genquant_0 \tr_0 \ldots \genquant_{k-1}\tr_{k-1} \A \left(\bigwedge_{j=0}^{k-1} \G(\altprop_{\tr_j} \leftrightarrow \neg\X\altprop_{\tr_j})\right) \rightarrow \psi\] is satisfiable, i.e., we universally quantify a trajectory~$t(0)t(1)t(2)\cdots$ and then require that if each $\tr_j$ is in each~$t(i)$, then $\psi$ has to hold. 
Then, the semantics of the temporal operators in \ahltl and \hyltl coincide.

For the upper bound, we reuse the encodings developed in the proof of Theorem~\ref{thm-Esat} and construct a formula~$\alpha(x)$ of arithmetic with a single free first-order variable~$x$ so that we have for all \aahltl sentences~$\phi$: $\natsstruct \models \alpha(\encode{\phi})$ if and only if $\phi$ is satisfiable. 
Intuitively, the formula should express, for a given~$\encode{\phi}$, that there is (the encoding of) a countable model, there are (encodings of) Skolem functions for the existentially quantified variables in $\phi$, \emph{for all}  (encodings of) trajectories~$t$, there is a function~$E$ from $\nats^4 $ to $\nats$ such that for all trace assignments~$\traceasg$ (encoded as some~$a_t \in \nats$) that are consistent with the Skolem functions, the function~$a_p,i,p \mapsto E(a_t,a_p,i,p)$ is the $(\psi, \traceasg,t)$-expansion and we have $E(a_t, a_0, 0, p^*) = 1$.
Here, $a_p$ encodes again a pointer assignment, $p$ encodes a quantifier-free subformula of $\phi$, and $p^*$ is the encoding of the maximal quantifier-free subformula of $\phi$.

However, this naive formula involves two  (second-order) quantifier alternations, i.e., it only shows membership in $\Sigma_3^1$.
To improve this, recall that the expansion is unique for a given trajectory and trace assignment. 
Hence, we can alternatively quantify $E$ universally and then require that the function~$a_p,i,p \mapsto E(a_t,a_p,i,p)$ is the $(\psi, \traceasg,t)$-expansion then we must have $E(a_t, a_0, 0, p^*) = 1$.
The resulting formula has only one (second-order) quantifier-alternation, which shows that \aahltl satisfiability is indeed in $\Sigma_2^1$.
\end{proof}

\section{Model-Checking}

In this section, we study the model-checking problem, which turns out to be harder than satisfiability. 

\begin{theorem}
\label{thm-mc}    
Model-checking \ahltl is equivalent to truth in second-order arithmetic. The lower bound holds for \aahltl and \eahltl.
\end{theorem}

\begin{proof}
For the upper bound, we present a polynomial-time translation mapping pairs~$(\TS, \phi)$ of transition systems and \ahltl sentences to sentences~$\phi'$ of second-order arithmetic such that $\TS\models \phi$ if and only if $\natsstruct\models\phi'$.
To this end, we capture the semantics of \ahltl in second-order arithmetic.
At first glance, this approach bears a resemblance to the upper bound proofs for \ahltl satisfiability (Theorem~\ref{thm-Esat} and Theorem~\ref{thm-Asat}).
However, there we \emph{only} had to handle countable sets of traces. 
Here, as $\traces(\TS)$ may be uncountable, we have to generalize the techniques to be able to quantify traces, trace assignments, pointer assignments, and trajectories. We begin by introducing their encodings.

So, let us fix $\TS$ and $\phi$ as above.
To encode traces of $\TS$, we fix a bijection~$h_\propos\colon \ap\rightarrow \nats$ and use Cantor's pairing function~$\pair \colon \nats\times\nats\rightarrow\nats$ defined as $\pair(i,j) = \frac{1}{2}(i+j)(i+j+1)+j$, which is a bijection that can be implemented in first-order arithmetic.

Then, we encode a trace~$\tra = \tra(0)\tra(1)\tra(2)\cdots$ over $\ap$ by the set
\[
S_\tra = \set{ \pair(i, h_\propos(\prop)) \mid i \in\nats \text{ and } \prop \in \tra(i) } \subseteq
\nats.\]
Note that $\tra\mapsto S_\tra$ is injective, but not every set encodes a trace of $\TS$. 
But there is a formula~$\alpha_{trc}(X)$ of second-order arithmetic with a single free second-order variable~$X$ such that $\natsstruct\models\alpha_{trc}(X)$ if and only if $X$ encodes a trace of $\TS$~\cite[Proof of Theorem~11]{frz26}, which relies on the fact that each transition system uses only finitely many propositions.

To encode trace and pointer assignments, we fix a bijection~$h_\varis \colon \vars\rightarrow \nats$. 
Then, we encode a trace assignment~$\traceasg$ by the set
\[
S_{\traceasg} = \set{ \pair(h_\varis(\tr), n) \mid \tr \in \dom{\traceasg} \text{ and } n \in S_{\traceasg(\tr)} } \subseteq \nats.
\]
Similarly, we encode a pointer assignment~$\pointerasg$ by the set
\[
S_{\pointerasg} = \set{\pair(h_\varis(\tr),\pointerasg(\tr)) \mid \tr \in \dom{\pointerasg}} \subseteq \nats.
\]
Again, the functions~$\traceasg\mapsto S_{\traceasg}$ and $\pointerasg \mapsto S_{\pointerasg}$ are injective, but not every set encodes a trace (pointer) assignment.
However, one can write second-order formulas~$\alpha_{ta}(X)$ and $\alpha_{pa}(X)$ that hold in $\natsstruct$ if and only if $X$ encodes a trace assignment (mapping to traces of $\TS)$ respectively a pointer assignment. The first formula relies on $\alpha_{trc}$.
Furthermore, one can write a formula~$\alpha_{empty}(X)$ with a single free second-order variable that is satisfied in $\natsstruct$ if and only $X$ encodes the trace assignment with empty domain.
Also, there is a formula~$\alpha_{init}(X)$ with a single free second-order variable that holds in $\natsstruct$ if and only if $X$ encodes the initial pointer assignment mapping every $\tr\in\vars$ to zero.

Finally, we encode a trajectory~$t = t(0)t(1)t(2)\cdots$ for the maximal quantifier-free subformula~$\psi$ of $\phi$ (i.e., each $t(i)$ is a subset of the variables occurring in $\psi$) by the set
\[
S_t = \set{ \pair(i, h_\varis(\tr)) \mid i \in\nats \text{ and } \tr\in t(i) } \subseteq \nats.
\]
Once more, the function~$t \mapsto S_t$ is injective, not every set encodes a trajectory for $\psi$, but there is a second-order formula~$\alpha_{trj}(X)$ that holds in $\natsstruct$ if and only if $X$ encodes such a trajectory. 

With these encodings and auxiliary formulas, we can construct the following formulas:
\begin{itemize}
    \item $\alpha_{lo}(A_t,A_p,p,j)$ with free second-order variables~$A_t$ (encoding a trace assignment~$\traceasg$) and $A_p$ (encoding a pointer assignment~$\pointerasg$) and free first-order variables~$p$ (encoding the proposition~$\prop = h_\propos^{-1}(p)$) and $j$ (encoding the variable~$\tr = h_\varis^{-1}(j)$) such that $\natsstruct\models\alpha_{lo}(A_t,A_p,p,j)$ if and only if $\prop \in \tra(\pointerasg(\tr))$ for $\tra =  \traceasg(\tr)$, i.e., $\alpha_{lo}$ looks up whether $\prop$ holds on the trace assigned to $\tr$ at the position induced by $\pointerasg$. This captures exactly the semantics of atomic propositions.

    \item $\alpha_{up}(A_t, A_t', X, j)$ with free second-order variables~$A_t$, $A_t'$ (encoding trace assignments~$\traceasg$ and $\traceasg'$), and $X$ (encoding a trace~$\sigma$ of $\TS$), and a free  first-order variable~$j$ (encoding the variable~$\tr = h_\varis^{-1}(j)$) such that $\natsstruct \models \alpha_{up}(A_t, A_t', X, j)$ if and only if $\traceasg' = \traceasg[\tr \mapsto \sigma]$, i.e., $\alpha_{up}$ implements the update of trace assignments.

    \item $\alpha_{fw}(A_p, A_p',T,i,k)$ with free second-order variables~$A_p$, $A_p'$ (encoding pointer assignments~$\pointerasg$ and $\pointerasg'$), and $T$ (encoding a trajectory~$t$), and free first-order variables~$i$ (interpreted as a position on $t$) and $k$ such that $\natsstruct \models\alpha_{fw}(A_p, A_p',T,i,k)$ if and only if $\pointerasg' = \forward{\pointerasg}{t}{i}{k}$, i.e., $\alpha_{fw}$ implements the iterated forward operator on pointer assignments.

\end{itemize}

Now, we can translate \ahltl into second-order arithmetic. The translation depends on the fixed transition system~$\TS$ whose traces the quantifiers range over. 
Furthermore, each formula of second-order arithmetic constructed by the translation has at most three free second-order variables (encoding a trace assignment, a pointer assignment, and a trajectory) and at most one free first-order variable (encoding a position on the trajectory).
\begin{itemize}
    \item $\arith(\exists\tr.\psi') = \exists X \alpha_{trc}(X) \wedge \exists A_t' (\alpha_{up}(A_t,A_t', X, h_\varis(\tr)) \wedge \arith(\psi'))$, where $A_t'$ is the free variable of $\arith(\psi')$ encoding the trace assignment and $A_t$ is the free variable of $\arith(\exists\tr  \psi')$ encoding the trace assignment. The other free variables are the same for both $\arith(\exists\tr  \psi')$ and $\arith(\psi')$.
    
    \item $\arith(\forall\tr  \psi') = \forall X  \alpha_{trc}(X) \rightarrow \exists A_t'  (\alpha_{up}(A_t,A_t', X, h_\varis(\tr)) \wedge \arith(\psi'))$ with the same setup of free variables as in the previous case.

    \item $\arith(\E \psi') = \exists T  \exists A_p  \exists i \ \alpha_{trj}(T) \wedge \alpha_{init}(A_p) \wedge i=0 \wedge \arith(\psi')$ where $T$ is the free variable of $\arith(\psi')$ encoding the trajectory, $A_p$ is the free variable of $\arith(\psi')$ encoding the pointer assignment, and $i$ is the free variable of $\arith(\psi')$ encoding the position on the trajectory. Hence, $\arith(\E \psi')$ does not have these free variables. The other free variable encoding the trace assignment is the same for both $\arith(\exists\tr  \psi')$ and $\arith(\psi')$.
    
    \item $\arith(\A \psi') = \forall T  \exists A_p  \exists i\  \alpha_{trj}(T) \rightarrow (\alpha_{init}(A_p) \wedge i=0 \wedge \arith(\psi'))$ with the same setup of free variables as in the previous case.

    \item $\arith(\prop_\tr) = \alpha_{lo}(A_t, A_p, h_\propos(\prop),h_\varis(\tr))$, i.e., $A_t$ and $A_p$ are the free variables for the assignments.

    \item $\arith(\neg \psi') = \neg \arith(\psi')$. Here, the free variables of $\arith(\neg \psi')$ are the free variables of $\arith( \psi')$.

    \item $\arith(\psi'_1 \vee \psi'_2) = \arith(\psi'_1) \vee \arith(\psi'_2)$ where we assume w.l.o.g.\ that all free variables of $\arith(\psi'_1)$ are equal to the free variables of $ \arith(\psi'_2)$, which are then also the free variables of $\arith(\psi'_1 \vee \psi'_2)$.

    \item $\arith(\X\psi') = \exists i'   \exists A_p'\  i' = i+1 \wedge \alpha_{fw}(A_p, A_p',T,i,1)\wedge \arith(\psi') $ where $A_p'$ is the free variable of $\arith(\psi')$ encoding the pointer assignment and $A_p$ is the free variable of $\arith(\X\psi')$ encoding the pointer assignment, $i'$ is the free variable of $\arith(\psi')$ encoding the position and $i$ is the free variable of $\arith(\X\psi')$ encoding the position, and both formulas share the same free variables encoding the trace assignment and the trajectory (which is $T$).

    \item $\arith(\psi'_1 \U \psi'_2) = \exists k  \exists i' \exists A_p'\  i' = i+k \wedge \alpha_{fw}(A_p, A_p', T, i, k) \wedge \arith(\psi'_2) \wedge \forall k'\  k'<k \rightarrow \exists i'' \exists A_p''\  i'' = i+k' \wedge  \alpha_{fw}(A_p, A_p'', T, i, k') \wedge \arith(\psi'_1)$ where $A_p'$ ($A_p''$) is the free variable of $\arith(\psi'_1)$ (of $\arith(\psi'_2)$) encoding the pointer assignment and $A_p$ is the free variable of $\arith(\psi'_1 \U\psi'_2)$ encoding the pointer assignment, $i'$ ($i''$) is the free variable of $\arith(\psi'_1)$ (of $\arith(\psi'_2)$) encoding the position and $i$ is the free variable of $\arith(\psi'_1\U\psi'_2)$ encoding the position, and both formulas share the same free variables encoding the trace assignment and the trajectory (which is $T$).
    
\end{itemize}
Now, we define $\phi' = \exists A_p  \alpha_{empty}(A_p) \wedge \arith(\phi)$, where $A_p$ is the free variable of $\arith(\phi)$ encoding the trace assignment. 
Then, we indeed have $\natsstruct \models \phi'$ if and only if $\TS \models \phi$ as required, which can be shown by an induction over the construction of $\phi$.


For the lower bound, we first consider universal trajectory quantification: We present a polynomial-time translation mapping sentences~$\phi$ of second-order arithmetic to pairs~$(\TS, \phi')$ of transition systems~$\TS$ and \aahltl sentences~$\phi'$ such that $\natsstruct \models \phi$ if and only if $\TS \models \phi'$.
Intuitively, we will capture the semantics of arithmetic in \aahltl.

We begin by formalizing our encoding of natural numbers and sets of natural numbers using traces. 
Intuitively, a trace~$\tra$ over a set~$\ap$ of propositions containing the proposition~$\intprop$ encodes the set~$\set{n \in \nats \mid \intprop \in \tra(n)} \subseteq \nats$.
Thus, a trace~$\tra$ encodes a singleton set if it satisfies the \ltl formula~$(\neg \intprop) \U (\intprop \wedge \X\G\neg \intprop)$.
In the following, we use the encoding of singleton sets to encode natural numbers as well.
Obviously, every set and every natural number is encoded by a trace in that manner.
Thus, we mimic first- and second-order quantification by quantification over traces.
Later, we use an additional proposition~$\auxprop$ to implement multiplication.
Furthermore, we use a proposition~$\altprop$ whose truth value will alternate on every trace. 
This is used to apply Remark~\ref{remark-trajaccess} to \myquot{access} the trajectory.

Fix a transition system~$\TS$ 
whose set of traces is 
\begin{align*}
    \traces(\TS) = & \set{ \tra(0)\tra(1)\tra(2) \cdots \in (\pow{\set{\intprop, \altprop}})^\omega \mid \altprop \in \tra(i) \leftrightarrow \altprop \notin \tra(i+1) \text{ for all } i\in\nats } \cup\\
    &\set{ \tra(0)\tra(1)\tra(2) \cdots \in (\pow{\set{\auxprop, \altprop}})^\omega \mid \altprop \in \tra(i) \leftrightarrow \altprop \notin \tra(i+1) \text{ for all } i\in\nats }.
\end{align*}
We call traces in the first set 
\myquot{set traces} and traces in the second set 
\myquot{auxiliary traces}.
Note that $(\emptyset\set{\altprop})^\omega$ and $(\set{\altprop}\emptyset)^\omega$ are both set and auxiliary traces. This is inconsequential for our construction.




    
    


    
    


Let $\alpha_{alt}(\tr) = \G(\altprop_\tr \leftrightarrow \X\neg\altprop_\tr)$ and let $\traceasg$ be a trace assignment with $\traceasg(\tr) \in \traces(\TS)$. Then, we have $(\traces(\TS),\traceasg) \models \A \alpha_{alt}(\tr) \rightarrow \G\neg\auxprop_\tr$ if and only if $\traceasg(\tr)$ is a set trace. Similarly, we have $(\traces(\TS),\traceasg) \models \A \alpha_{alt}(\tr) \rightarrow \G\neg\intprop_\tr$ if and only if $\traceasg(\tr)$ is an auxiliary trace.

We can now begin with our translation from second-order arithmetic to \aahltl. To this end, we assume w.l.o.g.\ that $\phi$ is a sentence of second-order arithmetic in prenex normal form, say
\[
\phi = \genquant_0 \nu_0\genquant_1 \nu_1\cdots \genquant_{k-1} \nu_{k-1}\phi'
\]
where each $\genquant_j$ is in $\set{\exists,\forall}$, each $\nu_j$ is either a first- or second-order variable, and $\phi'$ is quantifier-free. 
For each such variable~$\nu_j$, we introduce a trace variable~$\tr_{j}$ and define
\[
\phi' = \genquant_0 \tr_{0}\genquant_1 \tr_{1}\cdots \genquant_{k-1} \tr_{{k-1}} \exists \tr_0^a \cdots \exists \tr_{k'-1}^a  \A \alpha_e \wedge (\alpha_a \rightarrow  \hyperize(\phi'))
\]
where $\alpha_e $, $\alpha_e $ and $\hyperize(\phi')$ are defined below and where $\tr_0^a,\ldots, \tr_{k'-1}^a$ is a collection of fresh trace variables, four for each atomic formula of the form~$y = y' \cdot y''$ in $\phi$, that we will use to implement multiplication in \aahltl.
In the following, all our explanations assume that the trajectory is universally quantified, as it is the case in $\phi'$.
To define $\alpha_e $, $\alpha_e $ and $\hyperize(\phi')$, we rely on the universal quantification of the trajectory, as it allows us to \myquot{select}, using guard formulas, which trajectory we use to mimic, e.g., the semantics of quantifier-free second-order arithmetic formulas.
We define 
\[\alpha_e = 
\bigwedge_j \alpha_{alt}(\tr_{j}^a) \rightarrow (\G\neg\intprop_{\tr_j^a}) \wedge 
\bigwedge_{j'}\alpha_{alt}(\tr_{j'}) \rightarrow(\G\neg \auxprop_{\tr_{j'}}) \wedge \bigwedge_{j''} \alpha_{alt}(\tr_{j''}) \rightarrow (\neg \intprop_{\tr_{j''}}) \U (\intprop_{\tr_{j''}} \wedge \X\G\neg \intprop_{\tr_{j''}})  \]
where $j$ ranges over indexes in $\set{0,1,\ldots, k'-1}$, $j'$ ranges over all these indexes with $\genquant_{j'} = \exists$ and $j''$ ranges over all these indexes with $\genquant_{j''} = \exists$ and $\nu_{j''}$ is a first-order variable, i.e., $\alpha_e$ requires that the auxiliary variables~$\tr_j^a$ are assigned to auxiliary traces, and that traces assigned to existentially quantified trace variables~$\tr_j$ copied from $\phi$ correctly mimic their original in $\phi$. 
Now, $\alpha_a$ is defined analogously, we just let $j'$ and $j''$ range over universally quantified variables and universally quantified first-order variables, respectively (and, for completeness, $j$ ranges over the empty set). In the following, we assume that $\alpha_e$ and $\alpha_a$ are satisfied.

Finally, $\hyperize$ is defined inductively as follows:
\begin{itemize}
        
    \item $\hyperize(\neg \psi) = \neg \hyperize(\psi)$.
    
    \item $\hyperize(\psi_1 \vee \psi_2) = \hyperize(\psi_1) \vee \hyperize(\psi_2)$.

    \item $\hyperize(\nu_j \in \nu_{j'}) = (\alpha_{alt}(\tr_{j}) \wedge \alpha_{alt}(\tr_{j'})) \rightarrow  \F(\intprop_{\tr_{j}} \wedge \intprop_{\tr_{j'}})$.

    \item $\hyperize(\nu_j < \nu_{j'}) = (\alpha_{alt}(\tr_{j}) \wedge \alpha_{alt}(\tr_{j'})) \rightarrow  \F( \intprop_{\tr_{j}} \wedge \X\F \intprop_{\tr_{j'}})$.

\end{itemize}

At this point, it remains to consider addition and multiplication.
Let $\traceasg$ be a trace assignment that maps some trace variables~$\tr_{j},\tr_{j'},\tr_{j''}$ to set traces~$\tra,\tra',\tra''$ encoding singleton sets~$\set{n}, \set{n'}, \set{n''}$, respectively.
Our goal is to write quantifier-free formulas~$\hyperize(\nu_j = \nu_{j'} + \nu_{j''})$ and $\hyperize(\nu_j = \nu_{j'} \cdot \nu_{j''})$ with free variables~$\tr_{j}, \tr_{j'}, \tr_{j''}$ such that $\traceasg \models \A \hyperize(\nu_j = \nu_{j'} + \nu_{j''})$ if and only if $n = n' + n''$ 
and such that $\traceasg \models \A \hyperize(\nu_j = \nu_{j'} \cdot \nu_{j''})$ if and only if
$n = n' \cdot n''$ (note satisfaction of these formulas only depends on a trace assignment, not a set of traces, so we drop it).

For addition, we define
\begin{align*}
\hyperize(\nu = \nu_{j'} + \nu_{j''}) = 
\left[\left(
(
\altprop_{\tr_{j''}} \leftrightarrow \X\altprop_{\tr_{j''}}) \U (\intprop_{\tr_{j'}} \wedge \alpha_{alt}(\tr_{j''})) \right)\wedge 
\alpha_{alt}(\tr_{j'}) \wedge \alpha_{alt}(\tr_{j})
\right]
\rightarrow \F(\intprop_{\tr_{j''}} \wedge \intprop_{\tr_{j}}).
\end{align*}
Intuitively, if time is frozen on the trace assigned to $\tr_{j''}$ until $\intprop$ is encountered on the trace assigned to $\tr_{j'}$ (and progresses normally otherwise), then $n$ must satisfy $n = n' + n''$ as required.

To conclude, we consider multiplication, which is more involved than addition.
Our construction here is inspired by a similar construction developed for \hltlc~\cite{genhyltlsc}, but we need to replace contexts by trajectories.

We consider four different cases.
If $n' = 0$ or $n''=0$, then we must have $n = 0$ as well. This is captured by the formula~$\psi_1 = (\intprop_{\tr_{j'}} \vee \intprop_{\tr_{j''}}) \wedge \intprop_{\tr_{j}}$.
    Further, if $n' = n'' = 1$, then we must have $n = 1$ as well. This is captured by the formula~\[
    \psi_2 = (\alpha_{alt}(\tr_{j}) \wedge \alpha_{alt}(\tr_{j'}) \wedge \alpha_{alt}(\tr_{j''})) \rightarrow \X (\intprop_{\tr_{j'}} \wedge \intprop_{\tr_{j''}} \wedge \intprop_{\tr_{j}}).\]

Next, let us consider the case~$0 < n' \le n''$ with $n'' \ge 2$.
    Let $z \in \nats\setminus\set{0}$ be minimal with 
    \begin{equation}
    \label{eq}
        z \cdot (n'' - 1) = z' \cdot n'' - n'
    \end{equation}
    for some $z' \in \nats\setminus\set{0}$.
    It is easy to check that $z = n'$ is a solution of Equation~(\ref{eq}) for $z' = n'$.
    Now, consider some $0 < z < n'$ to prove that $n'$ is the minimal solution: 
    Rearranging Equation~\ref{eq} yields
    $    - z + n' = (z' - z)\cdot n''
    $, 
    i.e., $-z + n'$ must be a multiple of $n''$ (possibly $0$).
    But $0 < z < n'$ implies $0 < n' - z < n' \le n''$, i.e., $-z + n'$ is not a multiple of $n''$. Hence, $z = n'$ is indeed the smallest solution of Equation~(\ref{eq}).  
    So, for the minimal such $z$ we have $z \cdot (n'' -1) + n' = n' \cdot n''$, i.e., we have expressed multiplication of $n'$ and $n''$.
Our goal is to implement this reasoning in \aahltl. 

To this end, let $\tr^a$ and ${\tr^a}'$ be two of the fresh trace variables quantified in $\phi'$ for the purpose of implementing multiplication.
The formula
\[
(\alpha_{alt}(\tr^a) \wedge \alpha_{alt}({\tr^a}')) \rightarrow \left[ \auxprop_{\tr^a} \wedge \G\F\auxprop_{\tr^a} \wedge \G\F\neg \auxprop_{\tr^a}\wedge \G(\auxprop_{\tr^a} \leftrightarrow \auxprop_{{\tr^a}'})  \right]
\]
expresses that there are $m_j > 0$ for all $j$ such that the traces assigned to $\tr^a$ and ${\tr^a}'$ are, after projecting away $\altprop$, both of the form
\[
\set{\auxprop}^{m_0}\emptyset^{m_1}\set{\auxprop}^{m_2}\emptyset^{m_3}\set{\auxprop}^{m_4}\emptyset^{m_5}\cdots .
\]
Then, the formula
\[\left[
 \alpha_{alt}(\tr^a) \wedge 
(\auxprop_{{\tr^a}} \wedge (\altprop_{{\tr^a}'} \leftrightarrow \X \altprop_{{\tr^a}'}))\U (\neg \auxprop_{{\tr^a}} \wedge \alpha_{alt}({\tr^a}')) 
\right]
\rightarrow 
\auxprop_{{\tr^a}}\U (\neg \auxprop_{{\tr^a}} \wedge \G(\auxprop_{\tr^a} \leftrightarrow \neg \auxprop_{{\tr^a}'}))
\]
expresses that all $m_j$ are equal: The antecedent of the implication requires that time progresses normally on $\tr^a$, but is frozen on ${\tr^a}'$ until $\auxprop$ does not hold for the first time in the trace assigned to ${\tr^a}$. From there onward, time also progresses normally on ${\tr^a}'$.
In this situation, the until in the consequent then updates the pointer of ${\tr^a}$ to the last position of the first~$\set{\auxprop}$-block in the trace assigned to ${\tr^a}'$ and the pointer of ${\tr^a}'$ to zero. These positions are marked by \myquot{$1$} in Figure~\ref{fig-per} which illustrates the construction. From these positions onward, the always compares the pairs of positions connected by the diagonal arrows, thereby ensuring that the blocks all have the same length.

  \begin{figure}
    \centering
        \begin{tikzpicture}[thick,yscale=1.3]

        \node at (-.5,1) {$\tr^a$};
        \node at (-.5,0) {${\tr^a}'$};
        \def\y{.4}
        \node[fill=gray!20, circle, minimum size =12,inner sep = 0] at (3,1+\y) {1};
        \node[fill=gray!20, circle, minimum size =12,inner sep = 0] at (0,-\y) {1};
        
        \draw[->, > = stealth] (0,1) -- (12,1);
        \draw[->, > = stealth] (0,0) -- (12,0);

        \foreach \i in {0,1,...,11}{
    \draw (\i,-.1) -- (\i, .1);
    \draw (\i,.9) -- (\i, 1.1);
  }

        \foreach \i in {0,1,2,6,7,8}{
\draw[fill,red] (\i,.1) circle (.03);
\draw[fill,red] (\i,-.1) circle (.03);

\draw[fill,red] (\i,.9) circle (.03);
\draw[fill,red] (\i,1.1) circle (.03);
}

\begin{scope}
    \clip(0,0) rectangle (12,1);
  \foreach \i in {0,1,...,11}{
    \draw[<->,> = stealth,thin] (\i, .15) -- (\i+3,.85);
  }
\end{scope}

        \end{tikzpicture}
        \caption{The formula ensuring that the traces assigned to ${\tr^a}$ (and ${\tr^a}'$) are periodic. Here, \myquot{\,\raisebox{-0.25ex}{\begin{tikzpicture}[thick]
\protect\draw(0,-.1) -- (0, .1);
            \protect\draw[fill,red] (0,.1) circle (.03);
\protect\draw[fill,red] (0,-.1) circle (.03);
        \end{tikzpicture}}\,} (\myquot{\,\raisebox{-0.25ex}{\begin{tikzpicture}[thick]
\protect\draw(0,-.1) -- (0, .1);
        \end{tikzpicture}}\,}) denotes a position where $\auxprop$ holds (does not hold).}
        \label{fig-per}
        
    \end{figure}

Let $\alpha_{per}(\tr^a,{\tr^a}')$ be the conjunction of the two formulas enforcing that the traces assigned to ${\tr^a}$ and ${\tr^a}'$ are periodic.
Furthermore, let $\tr^a_{j'}, \tr^a_{j'+1}, \tr^a_{j'+2}, \tr^a_{j'+3} $ be the auxiliary variables designated to implement the multiplication~$\nu_j = \nu_{j'} + \nu_{j''}$ in $\phi$.
We use the formulas~$\alpha_{per}(\tr^a_{j'}, \tr^a_{j'+2})$ and $\alpha_{per}(\tr^a_{j'+1},\tr^a_{j'+3})$ as introduced above to ensure that $\tr^a_{j'}$ and $\tr^a_{j'+1}$ are periodic and then discard $\tr^a_{j'+2}$ and $ \tr^a_{j'+3}$, as they are only used in these formulas.
Then, we ensure that $\tr^a_{j'}$ has period~$n''$ and that $\tr^a_{j'+1}$ has period~$n''-1$ and then implement the reasoning behind Equation~\ref{eq}.
To this end, consider the formula
\begin{align*}
\psi_3 = &(\alpha_{alt}(\tr_{j'}) \wedge \alpha_{alt}(\tr_{j''})) \rightarrow (\X\F(\intprop_{\tr_{j'}} \wedge\F\intprop_{\tr_{j''}}) \wedge \X\X\F\intprop_{\tr_{j''}} )  \wedge\\
&\alpha_{per}(\tr^a_{j'}, \tr^a_{j'+2}) \wedge \alpha_{per}(\tr^a_{j'+1},\tr^a_{j'+3})\wedge\\
&(\alpha_{alt}(\tr_{j''}) \wedge \alpha_{alt}(\tr_{j'}^a)) \rightarrow \auxprop_{\tr_{j'}^a}\U(\neg\auxprop_{\tr_{j'}^a} \wedge \intprop_{\tr_{j''}})\wedge\\
&(\alpha_{alt}(\tr_{j''}) \wedge \alpha_{alt}(\tr_{j'+1}^a)) \rightarrow \auxprop_{\tr_{j'+1}^a}\U(\neg\auxprop_{\tr_{j'+1}^a} \wedge \X\intprop_{\tr_{j''}})\wedge\\
& \Bigg[ \Bigg. \left( 
\alpha_{alt}(\tr_{j'}^a) \wedge \alpha_{alt}(\tr_{j}) \wedge \alpha_{alt}(\tr_{j'}) \wedge \left( \left[ 
\altprop_{\tr_{j'+1}^a} \leftrightarrow \X\altprop_{\tr_{j'+1}^a}
\right] \U \left[
\intprop_{\tr_{j'}} \wedge \alpha_{alt}(\tr_{j'+1}^a)
\right]\right)
\right) \rightarrow\\
&\qquad\F(\intprop_{\tr_{j'}} \wedge (\neg \alpha_{algn})\U (\alpha_{algn} \wedge \X\intprop_{\tr_{j}})) \Bigg.\Bigg]
\end{align*}
where $\alpha_{algn} = (\auxprop_{\tr^a_{j'}} \leftrightarrow \neg\X\auxprop_{\tr^a_{j'}}) \wedge (\auxprop_{\tr^a_{j'+1}} \leftrightarrow \neg\X\auxprop_{\tr^a_{j'+1}})$ holds if the pointers for $\tr^a_{j'}$ and $\tr^a_{j'+1}$ point to the ends of a block in the respective trace.

Intuitively, the formula expresses the following:
\begin{itemize}
    
    \item The first line holds if we are indeed in the case~$0 < n' \le n''$ with $n'' \ge 2$.
    
    \item The second line holds if the auxiliary traces are indeed periodic and the third and fourth ones ensure that they have the right periods, i.e., $\tr^a_{j'}$ has period~$n''$ and that $\tr^a_{j'+1}$ has period~$n''-1$.
    
    \item The fifth line, which is the antecedent of an implication, holds if time progresses normally on $\tr_{j'}^a$, $\tr_{j'}$, and $\tr_{j}$, and if time is frozen on $\tr_{j'+1}^a$ until the (unique) position on the trace assigned to $\tr_{j'}$ is reached where $\intprop$ holds, from which point onward, time on $\tr_{j'+1}^a$ also progresses normally.

    \item If the antecedent holds, then the consequent in the sixth line implements the reasoning of Equation~\ref{eq} as follows (see also Figure~\ref{fig-mult}):
    \begin{itemize}
        \item The eventually updates the pointers of $\tr_{j}$ and $\tr_{j'}^a$ to $n'$ (the unique position where $\intprop$ holds on $\tr_{j'}$), while the pointer of $\tr_{j'+1}^a$ stays zero. These positions are marked with \myquot{$1$} in the figure.

        \item Then, the until operator relates positions~$i+n'$ on $\tr_{j'}^a$ and $i$ on $\tr_{j'+1}^a$ as indicated by the diagonal lines in the figure. Hence, it updates the pointer of $\tr_{j'+1}^a$ to $z \cdot (n''-1) -1$ for the smallest $z>0$ such that $z \cdot (n''-1) = z' \cdot n'' -n'$ for some $z'$. Accordingly, the pointers of $\tr_{j'}^a$ and $\tr_{j}$ are updated to $z \cdot (n''-1) -1 + n'$, as time progresses normally on these traces after the position~$n'$, as required by the antecedent. These positions are marked with \myquot{$2$} in the figure.

        \item Finally, the next increments the pointer of $\tr_j$ by one (this position is marked with \myquot{$3$}). Here, $\intprop$ must hold on the trace assigned to $\tr_j$.
    \end{itemize}
\end{itemize}
As argued above, the minimal~$z$ satisfying the above requirement is equal to $n'$, i.e., the pointer of $\tr_j$ after the application of the next operator is then equal to $z \cdot (n''-1) -1 + n' +1  =  n' \cdot n''$.
Hence, we have indeed implemented multiplication for the case~$0 < n' \le n''$ with $n'' \ge 2$.

        \begin{figure}[t]
        \centering
        \begin{tikzpicture}[thick,yscale=1.3]

    \def\y{1}
        \node at (-.5,\y) {$\tr_{j'}^a$};
        \node at (-.5,0) {$\tr_{j'+1}^a$};
        \node at (-.5,-\y) {$\tr_{j}$};

        \draw[->, > = stealth] (0,\y) -- (13,\y);
        
        \draw[->, > = stealth] (0,0) -- (13,0);

        \draw[->, > = stealth] (0,-\y) -> (13,-\y);
        
        \foreach \i in {0,.5,...,3,7,7.5,...,10}{
        \draw[fill,red] (\i,\y+.1) circle (.03);
        \draw[fill,red] (\i,\y+-.1) circle (.03);
        }

        \foreach \i in {0,.5,...,2.5,6,6.5,...,8.5, 12,12.5}{
        \draw[fill,red] (\i,.1) circle (.03);
        \draw[fill,red] (\i,-.1) circle (.03);
        }

        \foreach \i in {0,.5,...,12.5}{
    \draw (\i,-.1) -- (\i, .1);
    \draw (\i,\y-.1) -- (\i, \y+.1);
    \draw (\i,-\y-.1) -- (\i, -\y+.1);
    \pgfmathtruncatemacro{\result}{2*\i}
        \node at (\i,-1.25) {\scriptsize \result};
  }

      \foreach \i in {0,.5,...,8.5}{
    \draw[<->,> = stealth,thin] (\i, .15) -- (\i+1.5,\y-.15);
  }

        \node[fill=gray!20, circle, minimum size =12,inner sep = 0] at (1.5,\y+.4) {\footnotesize 1};
        \node[fill=gray!20, circle, minimum size =12,inner sep = 0] at (0,.4) {\footnotesize 1};
        \node[fill=gray!20, circle, minimum size =12,inner sep = 0] at (1.5,-\y+0.4) {\footnotesize 1};
        
        \node[fill=gray!20, circle, minimum size =12,inner sep = 0] at (8.5,.4) {\footnotesize 2};
        \node[fill=gray!20, circle, minimum size =12,inner sep = 0] at (10,\y+.4) {\footnotesize 2};
        \node[fill=gray!20, circle, minimum size =12,inner sep = 0] at (10,-\y+.4) {\footnotesize 2};
        
        \node[fill=gray!20, circle, minimum size =12,inner sep = 0] at (10.5,-\y+.4) {\footnotesize 3};

        \end{tikzpicture}
        \caption{The formula~$\psi_3$ implementing multiplication, for $n_1 = 3$ and $n_2 = 7$, i.e., $\tr_{j'}^a$ has period~$7$ and $\tr_{j'+1}^a$ has period~$6$. Here, \myquot{\,\raisebox{-0.25ex}{\begin{tikzpicture}[thick]
\protect\draw(0,-.1) -- (0, .1);
            \protect\draw[fill,red] (0,.1) circle (.03);
\protect\draw[fill,red] (0,-.1) circle (.03);
        \end{tikzpicture}}\,} (\myquot{\,\raisebox{-0.25ex}{\begin{tikzpicture}[thick]
\protect\draw(0,-.1) -- (0, .1);
        \end{tikzpicture}}\,}) denotes a position where $\auxprop$ holds (does not hold).}
        \label{fig-mult}
        
    \end{figure}

Finally, the construction for the last case~$0 < n'' < n'$ is analogous. Let the resulting formula be $\psi_4$.
Then, define
$\hyperize(\nu_j = \nu_{j'} \cdot \nu_{j''}) = \psi_1 \vee \psi_2 \vee \psi_3 \vee \psi_4$.
An induction over the construction of $\phi$ shows that we have $\natsstruct \models \phi$ if and only if $\TS \models \phi'$, where $\TS$ is the transition system introduced above.

Finally, let us consider \eahltl: As second-order arithmetic is closed under negation, checking whether a given sentence~$\phi$ holds in $\natsstruct$, i.e., truth in second-order arithmetic, is equivalent to checking whether $\neg \phi$ does \emph{not} hold in $\natsstruct$.
Thus, applying the reduction above and pushing the negation over the trace and trajectory quantifiers shows that truth in second-order arithmetic can be reduced to model-checking $\TS$ against \eahltl sentences.
\end{proof}

\section{Finite-state Satisfiability}

In this section, we consider the finite-state satisfiability problem for \ahltl.
As model-checking is hard for a fixed transition system, one can easily adapt the lower bound proof to show that finite-state satisfiability is also as hard as truth in second-order arithmetic.
Similarly, the upper bound can also be adapted, as one can quantify (encodings of) finite transition systems in second-order arithmetic.
This is in line with similar results for, e.g., \hyqptl~\cite{hyperqptlcomplexity}, and proven along the same lines.

\begin{theorem}
\label{thm-fssat}    
\ahltl finite-state satisfiability is equivalent to truth in second-order arithmetic. The lower bound holds for \aahltl and \eahltl.
\end{theorem}

\begin{proof}
Recall that we have reduced \ahltl model-checking to truth in second-order arithmetic by encoding the semantics of \ahltl in second-order arithmetic, using the formula~$\alpha_{trc}$ that checks whether a set~$X$ encodes a trace of a fixed transition system (see the proof of Theorem~\ref{thm-mc}). 
In second-order arithmetic, one can also quantify (an encoding of) a finite transition system and then check whether that transition system satisfies a given \ahltl sentence (cp.~\cite[Theorem~6.3]{genhyltlsc}).
Thus, \ahltl finite-state satisfiability can also be reduced to truth in second-order arithmetic.

For the lower bound, recall that we have reduced truth in second-order arithmetic to model checking \aahltl with respect to a fixed transition system~$\TS$ (see the proof of Theorem~\ref{thm-mc}). 
The set of traces of $\TS$ can be expressed in \aahltl, i.e., we can write a \aahltl sentence~$\alpha_m$ (using additional propositions) such that the projection of every model of $\alpha_m$ to $\set{\intprop,\auxprop,\altprop}$ is $\traces(\TS)$ (cp.~\cite[Theorem~6.1]{genhyltlsc}).
Thus, the finite-state satisfiability of $\alpha_m \wedge \phi$ (brought into prenex normal form (see Remark~\ref{remark-boolcombs})) is equivalent to model-checking $\TS \models \phi$, which is as hard as truth in second-order arithmetic.
\end{proof}

\section{Conclusion}

We have studied the complexity of satisfiability, finite-state satisfiability, and model-checking for \ahltl and showed that the former problem is $\Sigma_1^1$-complete for existentially quantified trajectories and $\Sigma_1^1$-hard and in $\Sigma^1_2$ for universally quantified trajectories, while the latter two problems are equivalent to truth in second-order arithmetic, where the lower bounds hold for both types of trajectory quantification.
Hence, all three problems are highly undecidable. 

This work extends a line of research that has settled the complexity of synchronous hyperlogics like \hyltl~\cite{hypercomplexity}, \hyqptl~\cite{hyperqptlcomplexity}, and \sohyltl~\cite{frz26} as well as asynchronous ones like generalized \hyltl with stuttering and contexts and its fragments~\cite{genhyltlsc}.
In future work, we aim to extend this to the remaining logics mentioned in Figure~\ref{fig-asynchlogics}.
Also, one can extend \ahltl with arbitrary trajectory quantification~\cite{DBLP:conf/tacas/HsuSB21}. As trajectories are objects of type~$1$, this does not influence the complexity of model-checking and finite-state satisfiability, but we expect the complexity of satisfiability to increase. 